\documentclass{article}

\usepackage{amssymb}
\usepackage{amsmath}
\usepackage{amsthm}

\newtheorem{theorem}{Theorem}[section]
\newtheorem{lemma}[theorem]{Lemma}

\newtheorem{claim}[theorem]{Claim}
\theoremstyle{definition}

\newtheorem{conj}[theorem]{Conjecture}
\newtheorem{defin}[theorem]{Definition}
\newtheorem{remark}[theorem]{Remark}

\newcommand{\la}{\lambda}

\newcommand{\al}{\alpha}
\newcommand{\fii}{\varphi}

\newcommand{\R}{\mathbb{R}}
\newcommand{\Z}{\mathbb{Z}}

\newcommand{\C}{\mathbb{C}}

\newcommand{\ie}{i.~e.~}
\newcommand{\fla}[2][\la]{f_{#1}(#2)}

\begin{document}

\title{Quantum tunneling on graphs}

\author{Yong Lin, G\'abor Lippner, Shing-Tung Yau}

\maketitle

\abstract{We explore the tunneling behavior of a quantum particle on a finite graph, in the presence of an asymptotically large potential. Surprisingly the behavior is governed by the local symmetry of the graph around the wells.}

\section{Introduction}

Quantum tunneling is the physical phenomenon that a quantum particle can get to the other side of an energy barrier. This is of course not possible in classical mechanics. Even more counterintuitive, Oskar Klein in~\cite{Klein} observed that the higher the barrier is, the higher the probability is that the particle crosses the barrier. As the height goes to infinity, the barrier becomes invisible to the particle. In this paper we shall investigate to what extent and under what conditions does quantum tunneling  happen in the discrete setting. In particular we shall see whether Klein's paradox can be recovered.

The energy barrier is easiest realized by a potential that has two or more local minima, referred to as energy wells. In physics, the motion of a quantum particle on a manifold is governed by the Schr\"odinger equation:
\[ i \hslash \frac{\partial}{\partial t} \Psi = H \Psi \]
Here $\Psi$ is the quantum state of the particle and $H = \Delta - V$ where $\Delta$ is the Laplace operator and $V$ is a given potential. 
Since the Laplace operator has a natural discrete analogue on graphs, this makes it possible to define and analyize the discrete Schr\"odinger equation.

Let $G(X,E)$ be a finite connected graph. For any vertex $x\in X$ we denote its degree by $d_x$. We think of the vertices as the possible positions of a quantum particle. At any moment $t$ the state of the particle is given by a unit length vector $\fii_t \in \C^X, \sum_{x\in X} |\fii_t(x)|^2 = 1$. The quantity $|\fii_t(x)|^2$ is interpreted as the probability of the particle being in position $x$. Let $\Delta$ denote the symmetrized Laplace matrix of the graph, that is 
\begin{equation}
\Delta(x,y) =  \left\{ \begin{array}{rcl} 1 & \text{ if } & x=y \\ 
\frac{-1}{\sqrt{d_x d_y}} & \text{ if } & xy \in E \\
0 & \text{ otherwise. } &\end{array} \right. 
\end{equation}

Given a potential $V : X \to [0,\infty)$, the motion of the particle is governed by the discrete Schr\"odinger equation  \begin{equation} -i \frac{d}{dt} \fii_t = H \fii_t \end{equation} where $H = \Delta - V$ is the Hamiltonian. (Here and from now on we will use the same notation for a function on the set of vertices and the corresponding diagonal matrix.)  We want to investigate the asymptotic behavior of this quantum evolution, so we choose the potential to be $V = Q \cdot W$ where $W$ is a fixed
vector and $Q$ is a real number going to infinity.

 Simon (see~\cite{Simon}) carried out asymptotic analysis of double well potentials in Euclidean space for rather general potentials. 
We will consider potentials with two (double well) and three (triple well) global maxima. We will consider \textit{simple potentials} meaning that the only vertices where the potential is non-zero are the wells themselves, and \textit{general potentials}.
It turns out that whether tunneling happens at all is primarily governed by the degree of similarity of the neighborhoods of the wells, and  the speed of tunneling depends on the distance of the wells. More precisely:

\begin{defin} Let $PR(x,k)$ denote the probability that the simple random walk started at $x \in X$ returns to $x$ at time $k$. Given two vertices $x,y \in X$ we say that they are $m$-cospectral if the probability of return is the same for $x$ and $y$ up to time $m$, that is $PR(x,k)=PR(y,k)$ for every $0\leq k \leq m$.
The \textit{cospectrality} of $x$ and $y$ is the maximal $m$ for which they are $m$-cospectral. This will be denoted by $co(x,y)$.
\end{defin}

\begin{defin}
Given two states $x,y \in X$, let us start the particle in the pure state of $x$. We define the tunneling coefficient to be
\begin{equation}
 \liminf_{Q \to \infty} \sup_{t \in [0,\infty)} |\fii_t(y)|^2 = TC(x,y). 
\end{equation}

We distinguish three different behaviors: 
\textit{Perfect (asymptotic) tunneling} happens if $TC(x,y) = 1$, \textit{partial tunneling} happens if  $0 < TC(x,y) < 1$ and 
\textit{no tunneling} takes place if $TC(x,y) = 0$. 

We say that the tunneling from $x$ to $y$ happens within $f(Q)$ time if 
\begin{equation}
TC(x,y) = \liminf_{Q \to \infty} \sup_{t \in [0,f(Q)]} |\fii_t(y)|^2. 
\end{equation}

We will use the Bachmann-Landau notation to talk about tunneling times.
\end{defin}

\begin{theorem}\label{doublewelltheorem} In the case of a simple double-well potential where the graph-distance of the two wells $x,y$ is denoted by $d= d(x,y)$ the following is true between $x$ and $y$:
\begin{enumerate} 
\item if the two wells are $d$-cospectral ($co(x,y) \geq d$) then there is perfect asymptotic tunneling
\item if $co(x,y) = d-1$ then there is partial tunneling
\item if $co(x,y) < d-1$ then there is no tunneling.
\end{enumerate}
When there is tunneling, the tunneling time is $\Theta(Q^{d-1})$.
\end{theorem}

The triple-well setup turns out to be much more complicated. 
To avoid excessive case analysis, here we restrict our attention to sufficiently cospectral wells.
 Let $x,y,z \in X$ be three points such that $a = d(x,y) \leq b = d(x,z) \leq c = d(y,z)$. Assume that the wells are pairwise $c$-cospectral, and let the potential be $W(x) = W(y) = W(z) = 1$ and 0 everywhere else. For any two wells $u,v$ let us compute 
\begin{equation}\label{cuveq}
c_{uv} = \sum_{P : u \to v, |P| = d(u,v)}\prod_{j=0}^{|P|-1} 1/\sqrt{d_{x_j}d_{x_{j+1}}}
\end{equation}
where $P: u\to v$ denotes any path $u=x_0,x_1,\dots x_l = v$ in the graph and the length of this path is denoted by $|P| = l$.

\begin{theorem}\label{triplewelltheorem} With the above notation and assumptions
\begin{enumerate}
\item If $a < b \leq c$ then there is perfect tunneling between $x$ and $y$, but no tunneling to and from $z$. The tunneling time is $\Theta( Q^{a-1})$.
\item If $a = b < c$ and the particle is started from $x$ then it never tunnels perfectly to a single other state, but it does tunnel to a mixed state of $y$ and $z$. This happens on the scale of $t \sim Q^{a-1}$. 
\item If $a = b< c$ and the particle is started from $y$, then it tunnels completely to $z$ if and only if $c_{xy}= c_{xz}$. The tunneling time is $\Theta(Q^{a-1})$.  Note that in particular,  unlike in the first case, the particle \textit{can} tunnel from $y$ to $z$ despite the fact that $d(y,x) < d(y,z)$. Comparing to the double-well case, note that the speed of tunneling between $y$ and $z$ increased substantially after introducing the third well at $x$. 
\end{enumerate}
\end{theorem}

The precise behavior becomes even more subtle If the wells form an equilateral triangle.

\begin{theorem}\label{equilateraltheorem} Assume $a= b=c$. 
\begin{enumerate}
\item If $c_{xy} \neq c_{xz}$ then there is only partial tunneling from $y$ to $z$.
\item If $c_{xy} = c_{xz}$ and $\sqrt{c_{yz}^2 + 8c_{xz}^2}$ is irrational, then there is perfect tunneling from $y$ to $z$. The tunneling time $f$ has to satisfy $Q^{a-1} = o(f(Q))$.
\item If $c_{xy} = c_{xz}$ and  $\sqrt{c_{yz}^2 + 8c_{xz}^2}$ is rational, then even if there is perfect tunneling from $y$ to $z$, the tunneling time $f$ has to satisfy $Q^{a} = O(f(Q))$. (In particular this is the case if $c_{xy} = c_{xz} = c_{yz}$.)
\end{enumerate}
\end{theorem}

\begin{conj} \label{conj}
If $c_{xy} = c_{xz} \neq c_{yz}$ we believe that there is always perfect tunneling from $y$ to $z$.
\end{conj}

In view of the characterisation of the double-well case, there is a very surprising instability phenomenon in the triple-well case.

\begin{theorem}\label{instabthm} It is possible to construct a triple-well scenario where $TC(y,z) = 4/9$, but by modifying the graph arbitrarily far from the wells only a little bit (e.~g.~by adding a single edge) one can achieve $TC(y,z)=1$. Hence the tunneling coefficient is not predictable from any fixed neighborhood of the wells.
\end{theorem}

Finally we show that tunneling can happen in the presence of more general potentials then the singular ones. In particular perfect tunneling happens in the symmetric double-well case.

\begin{theorem}\label{generalpotentialthm}
Let $G(X,E)$ be a graph, $x,y \in X$ and $W: X \to \R$ which has its global maxima in $x$ and $y$. Let us further assume that there is an involution of $G$ that takes $x$ to $y$ and that preserves $W$. Then there is perfect asymptotic tunneling between $x$ and $y$ and the tunneling time is of the usual order: $Q^{d(x,y)-1}$.
\end{theorem}

\subsection{Intuitive explanation of tunneling}\label{sketchsec}

It is easy to see that if the inital state of the particle is $\fii=
\fii_0$ then the solution of the Schr\"odinger equation is given by
\begin{equation}\fii_t = e^{itH}\fii.\end{equation} This evolution (without the potential) is usually referred to as the continuous time quantum walk.  It has been first studied in~\cite{FG}. For an overview the reader is referred to~\cite{Kempe}. The main application of quantum walks is in quantum computing. Since the quantum walk "moves" faster than the classical random walk (see e.~g.~\cite{CFG}) it can be used to show the power of quantum computing over classical computers.

To better understand the behavior of the solution we decompose the state-space according to eigenstates. Let $n = |X|$ denote the number of vertices of $G$.
The matrix $H$ is symmetric hence it has a real spectrum with
eigenvalues $\la_1 \leq \la_2 \leq \la_3 \leq \dots \leq \la_n$. Let
$\psi_k$ denote the eigenvector corresponding to the eigenvalue
$\la_k$. Then $e^{itH}$ has the same eigenvectors but with eigenvalues
$e^{it\la_k}$. Thus if we write \begin{equation}\label{decompositioneq}\fii = \sum_{k=1}^n c_k \psi_k\end{equation}
then we get that \begin{equation}\label{fi_tdecompeq} \fii_t = \sum_{k=1}^n c_k e^{it\la_k} \psi_k.\end{equation}

If the potential well is deepest at the vertices $x,y \in X$ with
value $V(x) = V(y) = Q$ then from the physical analogy one expects the
two smallest (most negative) eigenvalues to be approximately $1-Q$,
and the corresponding two normalized eigenfunctions should be close to
$\psi_1(x) \approx \psi_1(y) \approx \psi_2(x) \approx -\psi_2(y)
\approx 1/\sqrt{2}$ and $\psi_1(z) \approx \psi_2(z) \approx 0$ if $z
\in X \setminus \{x,y\}$. Then if the particle starts in $x$ we
have \begin{equation}\fii \approx \frac{\psi_1 + \psi_2}{\sqrt{2}}.\end{equation} Hence \begin{equation}\fii_t
\approx \frac{e^{it\la_2}}{\sqrt{2}}\left(e^{it(\la_1-\la_2)} \psi_1 +
\psi_2 \right)\end{equation} so at time $t = \pi /(\la_1 - \la_2)$ the state of
the particle is approximately given by $c(\psi_1 - \psi_2)$ so the
particle is almost surely in state $y$.

In the following sections we shall make precise estimates to 
be able to confirm or reject the predicted behavior of the
particle. 

\section{Asymptotics}

In order to make the intuitive explanation work in reality, we have to
understand asymptotic behavior of various quantities. The initial
state of our particle will be $\fii = \chi_x$, the characteristic
function of the vertex $x \in X$ which is one of the vertices where
the potential is minimal. Hence first of all we will be interested in
the eigenvector-decomposition $\fii = \sum_1^n c_k \psi_k$. To be able
to analyise this, we will first need to understand asymptotics of the
eigenvalues and eigenspaces of the Hamiltonian.

We shall always denote the increasing sequence of eigenvalues of $H =
\Delta - Q\cdot W$ by $\la_1 \leq \la_2 \leq \dots \la_n$, and the
corresponding normalized, pairwise orthogonal eigenvectors by
$\psi_1,\psi_2, \dots, \psi_n$. (If there are multiple eigenvalues,
the choice of $\psi_i$ might not be unique.)

\subsection{Eigenspaces}\label{eigenspacessec}

The spectrum of a matrix depends continuosly on the entries. Since $H/Q = \Delta/Q - W \to -W$ if $Q \to \infty$ we can observe the following: 

\begin{claim}
If the values of the unscaled potential $W$ are $w_1 \leq w_2 \leq \dots \leq w_n$ then for all $1\leq i \leq n$ we have \begin{equation}\lim_{Q \to \infty} \frac{\la_i}{Q} +w_i = 0.\end{equation}
\end{claim}

\begin{lemma}
For any $i$ the eigenfunction $\psi_i$ is asymptotically concentrated on those vertices of the graph on which the value of the unscaled potential is equal to $w_i$. That is $\psi_i(x) \to 0$ as $Q \to \infty$ unless $W(x) = w_i$.
\end{lemma}

\begin{proof} 
Denote by $\chi_x$ the characteristic vector of $x \in X$ and write
$\chi_x = \sum c_k \psi_k$. Then $\psi_i(x) = \langle \psi_i, \chi_x
\rangle = c_i$. We further have $H \chi_x = \sum \la_k c_k \psi_k$ and
hence \begin{equation}\Delta \chi_x = (H + QW)\chi_x = H \chi_x + QW(x)\chi_x = \sum
(\la_k + QW(x))c_k \psi_k.\end{equation} Dividing both sides by $Q$ and taking
the norm-squares we get that
\begin{equation} \frac{||\Delta||^2}{Q^2} \geq \frac{||\Delta \chi_x ||^2}{Q} = \sum_{k=1}^n \left( \frac{\la_k}{Q}+W(x)\right)^2 c_k^2 \geq 
\left(\frac{\la_i}{Q}+W(x)\right)^2 c_i^2.\end{equation}

The left hand side converges to zero, and if $W(x) \neq w_i$ then
$\la_i/Q  + W(x)$ converges to a non-zero constant, hence $c_i$ has to
converge to 0 as $Q$ goes to infinity, and this finishes the proof.
\end{proof}

A level set of the potential $W$ is a maximal subset $L \subset X$ on
which $W$ is constant. The preceeding lemma shows that for any such
level set there will be precisely $|L|$ eigenfunctions of the
Hamiltonian that become more and more concentrated on $L$, and
conversely, any function concentrated on $L$ will be more and more
composed only of these $|L|$ eigenfunctions.

The ``wells" are the vertices of the level set corresponding to the
largest value of $W$ (since we are subtracting $W$ this is where the
potential is actually minimal). As we are interested in the behavior
of a particle started from one of the wells, we need to understand the
asymptotic behavior of the eigenfunctions corresponding to this level
set. We know that they become concentrated on the wells and we already
understand the first order behavior of the corresponding
eigenvalues. However to study the tunneling phenomenon, we need a
subtler description of the eigenfunctions involved to see whether
tunneling actually happens or not. We also need a more precise
asymptotic on the eigenvalues to be able to estimate the tunneling
time.

\subsection{Eigenfunctions}

The method we use to understand eigenfunctions comes from the theory
of random walks. It is well known that for any nonempty subset $L
\subset X$ and a fixed function $f: L \to \R$ there is a unique
extension $f: X \to \R$ that is harmonic everywhere outside of
$L$. One possible way to do this extension is to start a random walk
from a vertex $x$ and define $f(x)$ to be the expected value of $f$ at
the point where the walk first hits $L$. This method can be extended
to the case, where we want $f$ to be an eigenfunction of the Laplacian
with eigenvalue $\la$ outside of $L$. Denote the vertices of the
random walk by $x = v_0, v_1,\dots, v_T$ where $T$ is the first time
when the walk enters the subset $L$. Then define $f(x)$ to be the
expected value of $f(v_T)/(1-\la)^T$.  The prescence of the potential
makes things even more complicated, but the construction below is
still inspired by the previous argument.

Let $L \subset X$ be the level set of the wells. We may assume without
loss of generality that $W|_L = 1$ and $0 \leq W|_{X \setminus L} <
1$. Fix an arbitrary function $f: L \to \R$ and extend it to $X
\setminus L$ the following way. Let us consider a walk $P = \langle
x_0, x_1, \dots, x_T \rangle$ of length $|P| = T$ on the graph
starting from a vertex $x=x_0 \in X$ and ending in $v = x_T \in L$
such that if $0 < j < T$ then $x_j \in X \setminus L$. Such a walk
will be referred to as walk from $x$ to $L$ and denoted by $P : x \to
L$ or $P : x \to v$ if we want to specify the endpoint. (Thus we allow
the $x\in L$ case. It is implicitly understood in the notation that
the walk does not visit $L$ except for the last vertex and maybe the
first.) For any such walk $P$ we define a weight by the formula
\begin{equation}s(P) = \prod_{j=0}^{|P|-1} \frac{1}{\sqrt{d_{x_j} d_{x_{j+1}}} (1- \la - QW(x_j))}.\end{equation}
Then we construct our extension of $f$ using the following expression. If $x \in X \setminus L$ then put
\begin{equation} \label{extensionformula}
f(x) = \sum_{P : x\to L} f(x_{|P|})s(P) = \sum_{v\in
  L} f(v) \sum_{P: x\to v} s(P).  
\end{equation}

First of all, it is easy to see that this infinte sum will be absolutely convergent for large $Q$ (and hence large $\la$). Since $W$ is strictly smaller than 1 outside of $L$, the terms $1/(1-\la-QW(x_j))$ will be uniformly bounded from above by some $c/Q$ as $Q$ goes to infinity. On the other hand, by induction on $k$ we have
\begin{equation}
\sum_{P: x\to v, |P| = k} \prod_{j= 0}^{k-1} 1/\sqrt{d_{x_j} d_{x_{j+1}}} \leq \sqrt{d_x d_v}.
\end{equation} 
Hence if we sort the terms according to the length of $P$, we get a power series in $1/Q$ (or $1/\la$) with bounded coefficients. This easily implies absolute convergence.

Let us consider now a walk $P : x \to L$ and denote $y = x_1$. Let
$R$ denote the walk $y,x_2,\dots, x_T$, \ie $R$ is obtained from $P$ by
removing the first vertex, and $P = xR$.  By definition $s(R) = s(P) \sqrt{d_x d_y}
(1-\la - QW(x))$, and obviously $|R| = |P|-1$. Thus we can partition all
paths $P: x\to L$ according to their second vertex (denoted by $x_1$)
to arrive at the following identity:
\begin{multline} f(x) = \sum_{xy \in E} \sum_{R: y\to L}
f(x_{|R|+1})s(xR) =\\= \frac{1}{1- \la - QW(x)}\sum_{xy \in
  E} \frac{1}{\sqrt{d_xd_y}}\sum_{R: y\to L}
f(x_{|R|+1})s(R) = \\ = \frac{1}{1 - \la - QW(x))} \sum_{xy
  \in E} \frac{f(y)}{\sqrt{d_xd_y}}.
\end{multline}
Multiplying both sides by $1-\la-QW(x)$ we see that $(Hf)(x) = \la f(x)$ hence $f$ is behaves like an eigenfunction of the Hamiltonian for every vertex in $X \setminus L$. We could actually use the same formula to ``extend" $f$ to $L$, only it is already defined there. Hence the condition that $f$ is really an eigenfunction for $H$ is equivalent to saying that the ``extension" agrees with the starting value. Hence for every $v \in L$ we get an equation 
\begin{equation} \label{eigenvalueeq}
f(v) =  \sum_{P : v\to L}f(x_T) s(P) = \sum_{w\in L} f(w) \sum_{P : v\to w}s(P).
\end{equation}

Let us denote by $Z_{vw}(\la,Q) = \sum_{P: v\to w} s(P)$ the sum of the weights for all paths from $v$ to $w$, as a function of $\la$ and $Q$. These functions form an $|L| \times |L|$ matrix $Z = (Z_{vw})$ and the conditions on the values $f|_L$ can be written simply as $f|_L = Z \cdot f|_L$. Hence we have reduced the problem of analyzing eigenfunctions of an $n \times n$ matrix to understanding eigenfunctions of an $|L| \times |L|$ matrix whose entries are functions. This isn't a huge gain in general, but for the particular cases we are interested in, it does help a lot, as we shall see.

\subsection{Symmetric double wells -- the proof of Theorem~\ref{generalpotentialthm}}

To illustrate the usefulness of the matrix $Z(\la, Q)$ we investigate the case when $W$ has 2 global maxima: $x$ and $y$, and there is an involution of $G$ that preserves $W$. In this case there is a bijection between paths from $x$ to $x$ and paths from $y$ to $y$. Also there is a bijection between paths from $x$ to $y$ and paths from $y$ to $x$. These bijections obviously preserve the weight $s(P)$. Hence $Z_{xx} = Z_{yy}$ and $Z_{xy} = Z_{yx}$. This way the eigenvectors of the matrix $Z$ are automatically $(1,1)$ and $(1,-1)$ with respective eigenvalues $Z_{xx} \pm Z_{xy}$. 
This means that the eigenvectors for $H$ are exactly as we have presumed in Section~\ref{sketchsec}. Hence there is perfect asymptotic tunneling, exactly as described there. To determine the tunneling time we need to understand that for fixed $Q$, how are the solutions of $Z_{xx}(\la_1, Q) + Z_{xy}(\la_1,Q) = 1$ and$Z_{xx}(\la_2, Q) - Z_{xy}(\la_2,Q) = 1$ relate to each other.
Let us introduce the auxiliary variable $h = (1-\la) / Q$. Now the weight of each path can be rewritten as 
\begin{equation} 
s(P) = \frac{1}{Q^{|P|}}\prod_{j=0}^{|P|-1} \frac{1}{\sqrt{d_{x_j}d_{x_{j+1}}} (h- W(x_j))}
\end{equation}
Let us write $s'_P = s'_P(h) = s(P)(h-1)Q^{|P|}$. Then, since we know that $h \to 1$ as $Q \to \infty$ and that $0 \leq W(v) < 1$ when $v$ is not a well, we get that $s'_P$ tends to some constant as $h \to 1$.
Now, after multiplying by $h-1$, the equations for the eigenvalues of $Z$ being equal to 1 can be written as 
\begin{equation}
h-1 = \sum_{P : x \to x} \frac{s'_P(h)}{Q^{|P|}} \pm \sum_{P : x\to y} \frac{s'_P(h)}{Q^{|P|}}
\end{equation}
Denoting the solutions of each version by $h_1,h_2$, taking the difference of the two equations and dividing by $h_1 - h_2$ we get
\begin{equation}
1 = \sum_{P : x \to x} \frac{s'_P(h_1)- s'_P(h_2)}{(h_1-h_2)Q^{|P|}} + \frac{Q^{-d(x,y)}}{h_1-h_2} \sum_{P : x\to y} \frac{s'_P(h_1)+s'_P(h_2)}{Q^{|P|-d}}.
\end{equation}
Since each $s'_P$ is clearly differentiable and the length of any $x\to x$ path is at least 2 by definition, the first sum on the right hand side goes to 0 as $Q \to \infty$. In the second sum all the $s'_P$'s converge to some constants, and hence all terms where $|P| > d(x,y)$ go to 0, while those finitely many terms where $|P| = d(x,y)$ converge to constants. Hence $Q^{d(x,y)}/(h_1-h_2)$ has to converge to a constant too. Thus we get $h_1 - h_2 \sim 1/Q^{d(x,y)}$ and hence $\la_1 - \la_2 \sim 1/Q^{d(x,y)-1}$. Combining this with the observations in Section~\ref{sketchsec} we get that the tunneling time is of order $Q^{d(x,y)-1}$. \qed

\section{Simple potential}

In this section we are going to analyze the two and three-well scenarios where the potential wells are ``singular", in the sense that $W(v) = 0$ if $v$ is not one of the wells and $W(v) = 1$ at each well.
In this case the entries of matrix $Z$ have particularly simple dependence on $Q$ which makes the analysis possible.  
First of all, the weight function on paths becomes 
\begin{equation} 
s(P) = \frac{1}{(1-\la)^{|P|}}\prod_{j=0}^{|P|-1} \frac{1}{\sqrt{d_{x_j}d_{x_{j+1}}}},
\end{equation} 
when the starting vertex of $P$ is not a well. For paths starting from the wells we get 
\begin{equation}
s(P) = \frac{1-\la}{1-\la-Q} \frac{1}{(1-\la)^{|P|}} \prod_{j=0}^{|P|-1} \frac{1}{\sqrt{d_{x_j}d_{x_{j+1}}}}.
\end{equation}
This finally gives 
\begin{equation}
Z_{vw}(\la,Q) = \frac{1-\la}{1-\la-Q} \sum_{k = 1}^{\infty} \frac{1}{(1-\la)^k} \sum_{\begin{array}{cc} P: v \to w \\ |P| = k \end{array}} \prod_{j=0}^{k-1} \frac{1}{\sqrt{d_{x_j}d_{x_{j+1}}}}.
\end{equation}
The last part of the formula is independent of $\la$ and $Q$. Hence we can introduce the notation 
\begin{equation} 
P_{vw}(k) = \sum_{\begin{array}{cc} P: v \to w \\ |P| = k \end{array}} \prod_{j=0}^{k-1} \frac{1}{\sqrt{d_{x_j}d_{x_{j+1}}}}
\end{equation} 
so we have 
\begin{equation}
Z_{vw}(\la,Q) =\frac{1-\la}{1-\la-Q} \sum_{k = 1}^{\infty} \frac{P_{vw}(k)}{(1-\la)^k} .
\end{equation}

\begin{claim}\label{returnprobclaim}
For any $x \in X$ the value of $P_{xx}(k) = PR(x,k)$ is exactly the probablility that the simple random walk started from $x$ returns after $k$ steps.
\end{claim}
\begin{proof}
For any closed path $P : x = x_0,x_1,\dots,x_T =x$ since the starting and endpoint are the same, we actually have
\begin{equation}
\prod_{j=0}^{T-1} \frac{1}{\sqrt{d_{x_j}d_{x_{j+1}}}} = \prod_{j=0}^{T-1} \frac{1}{d_{x_j}},
\end{equation}
and this is exactly the probability of the random walk traversing $P$. Hence the sum for all $P$ of length $k$ is exactly the probability of return in $k$ steps.
\end{proof}

We are interested in when 1 is an eigenvalue of the matrix $Z(\la,Q) = (Z_{vw}(\la,Q))$, but this is clearly equivalent to $(1-\la-Q)/(1-\la) = 1- Q/(1-\la)$ being an eigenvalue of the transformed matix $\tilde{Z}(\la)$ which consist of entries
\begin{equation} 
\tilde{(Z)}_{vw}= \sum_{k=1}^{\infty} \frac{P_{vw}(k)}{(1-\la)^k}. 
\end{equation}
 
Summarizing what we have computed until now: if $Q$ is large, then the Hamiltonian will have two eigenvalues close to $-Q$. The corresponding eigenfunctions restricted to the wells will be eigenfunctions of the matrix $\tilde{Z}$ with eigenvalue $1-Q/(1-\la)$. On the other hand if you start with an eigenfunction of $\tilde{Z}$ with eigenvalue $1-Q/(1-\la)$ then extending it by formula~(\ref{extensionformula}) we get an eigenfunction of $H$ with eigenvalue $\la$. We are in very good shape now, as for any eigenfunction of $\tilde{Z}(\la)$ there is a unique $Q$ that makes it into an eigenfunction with eigenvalue $1-Q/(1-\la)$ and the preceeding argument ensures that $Q \approx -\la$ if $\la$ is large enough. This way we can get rid of the $Q$ parameter and only work with $\tilde{Z}(\la)$. This matrix is symmetric, since for any path $P: v \to w$ the reverse $\bar{P}$ is a path from $w$ to $v$ with the same length and weight (the latter is so because we are using the symmetrized Laplacian).

\subsection{Double wells -- the proof of Theorem~\ref{doublewelltheorem}}

Let us start with the classical case, when there are two wells, $x$ and $y$ and the matrix $\tilde{Z}(\la)$ is $2 \times 2$. The first question we have to answer is whether there is tunneling or not. This depends on how the two eigenvectors of $\tilde{Z}(\la)$ behave. Let us denote the two coordinates of an eigenvector by $\fla{x}$ and $\fla{y}$, noting the dependence on $\la$. Since $\tilde{Z}_{xy}(\la) \neq 0$, neither coordinates of the eigenvector will be zero. Thus in both equations coming from~(\ref{eigenvalueeq}) we can divide by term on the left hand side and combine the two equations into one, obtaining
\begin{equation}
\tilde{Z}_{xx}(\la) + \frac{\fla{y}}{\fla{x}} \tilde{Z}_{xy}(\la) = \frac{\fla{x}}{\fla{y}}\tilde{Z}_{yx} (\la)+ \tilde{Z}_{yy}(\la).
\end{equation}

Using the symmetry of $\tilde{Z}$ and putting $\al = \al(f) = \fla{x}/\fla{y}$ we get 
\begin{equation}
\tilde{Z}_{xy}\left(\al -\frac{1}{\al}\right) = \tilde{Z}_{xx} - \tilde{Z}_{yy}. 
\end{equation}

Since our eigenvectors are normalized, we are only interested in the limiting behavior of $\al$ as $\la \to -\infty$. Since there are two eigenvectors and they are orthogonal, if the other eigenvector is $g$ then $\al(g) = -1/\al(f)$. Hence the limiting behavior of $\al(f)$ and $\al(g)$ together as an unordered pair is clearly determined by the behavior of $\al-1/\al$. Hence the only thing we have to understand is 
\begin{equation}
\lim_{\la \to -\infty} \left(\al - \frac{1}{\al}\right) = \lim_{\la \to -\infty} \frac{\tilde{Z}_{xx}(\la)-\tilde{Z}_{yy}(\la)}{\tilde{Z}_{xy}(\la)}. 
\end{equation}
Further simplifying notation by writing $t = 1/(1-\la)$ we have to analyze the limit 
\begin{equation}
\lim_{t \to +0}\frac{\sum_{k=2}^{\infty} (P_{xx}(k) - P_{yy}(k)) t^k}{\sum_{k=1}^{\infty}P_{xy}(k)t^k}
\end{equation}

The denominator's leading term is clearly $t^{d(x,y)}$ where $d(x,y)$ denotes the distance between $x$ and $y$ in the graph. Depending on the graph three different limiting behaviors are possible. 
If for some $k < d(x,y)$ the difference $P_{xx}(k) - P_{yy}(k)$ is non-zero (\ie the cospectrality of the wells is less than $d-1$), then the limit will be $\pm \infty$, hence the two eigenvectors tend to $(0,1)$ and $(1,0)$. In this case the initial state of the particle is almost an eigenstate itself, so the particle will remain very close to its initial state, there is no asymptotic tunneling at all. 

If for all $k \leq d(x,y)$ the difference $P_{xx}(k) - P_{yy}(k)$ is zero (\ie when the wells are $d$-cospectral), then the limit is 0, hence the two eigenvectors tend to $(1/\sqrt{2}, 1/\sqrt{2})$ and $(1/\sqrt{2},-1/\sqrt{2})$. This is exactly the setting of Section~\ref{sketchsec}. As we have seen in Section~\ref{eigenspacessec} all but two $c_k$'s in (\ref{decompositioneq}) converge to 0, and we have just shown that the remaining two converge to the same value $1/\sqrt{2}$. Hence according to the argument in Section~\ref{sketchsec} there is asymptotically perfect between $x$ and $y$.

Finally in the case when the first non-zero difference $P_{xx}(k)-P_{yy}(k)$ is for $k = d(x,y)$, that is, the cospectrality of the wells is exactly $d-1$, the limit will be a non-zero constant $c = (P_{xx}(d(x,y))-P_{yy}(d(x,y)))/P_{xy}(d(x,y))$. Hence the two eigenvectors will converge to $(a,b)$ and $(-b,a)$ where $a/b -b/a = c, a^2+b^2=1$. In this case $a \neq b$ and in (\ref{decompositioneq}) the two non-vanishing coefficients stablize to $a$ and $-b$. Hence we get 
\[ \fii_t(x) \approx e^{it \la_2}\left( e^{it(\la_1 - \la_2)} a^2 + b^2 \right).\] So maximal tunneling still occurs at time $t = \pi /(\la_1-\la_2)$ however it is not perfect tunneling. The probability of the state $x$ never goes asymptotically below $b^2$, and the probability of the state $y$ never goes above $a^2 = 1-b^2$.

To finish the proof of Theorem~\ref{doublewelltheorem} we have to estimate the tunneling time. From the arguments in Section~\ref{sketchsec} it is clear that the tunneling time is approximately $\pi/|\la_1-\la_2|$ where $\la_1$ and $\la_2$ are the eigenvalues corresponding to the two eigenfunctions concentrated on the wells. 
The difficulty is that until this point we ``fixed" $\la$, chose an eigenvector for $\tilde{Z}(\la)$. We didn't have to worry about the eigenvalue, because we could always find a suitable $Q$ for which everything worked in the end.
Of course since $\tilde{Z}$ has two eigenvalues, so we actually get two choices for $Q$ each of which gives one of the eigenvectors in return. Now we have to reverse things, and compute the two possible $\la$'s as a function of $Q$. To do this, first we describe how the two possible $Q$'s behave as the function of $\la$.

Let $\nu_{1,2}$ denote the two eigenvalues of $\tilde{Z}$. Then, since $\nu = 1-Q/(1-\la)$ has to hold, the possible choices for $Q$ are given by $Q_{1,2} = (1-\la)(1-\nu_{1,2})$. The $\nu_{1,2}$ are given as the solution of a quadratic polyomial, so we can write them down explicitly:
\begin{equation}
\nu_{1,2} = \frac{\tilde{Z}_{xx}+\tilde{Z}_{yy} \pm \sqrt{(\tilde{Z}_{xx} - \tilde{Z}_{yy})^2 + (2\tilde{Z}_{xy})^2}}{2},
\end{equation}
hence both roots are real and we can compute:
\begin{equation}\label{hlambda}
h(\la) = Q_1 - Q_2 = (1-\la)\left(\sqrt{(\tilde{Z}_{xx} - \tilde{Z}_{yy})^2 + (2\tilde{Z}_{xy})^2}\right).
\end{equation}
We have already seen that $Q_1$ is a power series of the form  $Q_1(\la) = \la + c_0 +c_1/\la + c_2/\la^2 + \dots$, and by the definition of $h$ we have $Q_2(\la) = Q_1(\la) - h(\la)$. Let us denote the inverses of $Q_1(\la)$ and $Q_(\la)$ by $\la_1(Q)$ and $\la_2(Q)$. We will see that $h(\la)$ is a power series in $1/\la$ with no constant term and then by the following lemma we get that $\la_2(Q) - \la_1(Q) = h(Q) + o(h(Q))$, so the tunneling time turns out to be $\pi/(h(Q)+o(h(Q)))$.

\begin{lemma}\label{inverselemma} Let $f(t) = t + c_0 + c_1/t + c_2/t^2+\dots$ and let $g(t) = f(t) - h(t)$ where $h(t)$ is a power series in $1/t$ with no constant term.
Then  $g^{-1}(t) = f^{-1}(t) + h(t) + o(h(t))$.
\end{lemma}

\begin{proof}
$g^{-1}(t) = f^{-1}(f(g^{-1}(t))) = f^{-1}(g(g^{-1}(t))+h(g^{-1}(t))) =f^{-1}(t+h(g^{-1}(t))) = f^{-1}(t) + f^{-1'}(t)h(g^{-1}(t)) + o(h(g^{-1}(t)))$. But $f^{-1'}(t) = 1 + o(1)$ and $g^{-1}(t) = t + O(1)$ hence $h(g^{-1}(t)) = h(t+O(1)) = h(t) + o(h(t))$. Putting these together we see that indeed $g^{-1}(t) =  f^{-1}(t) + h(t) + o(h(t))$.
\end{proof}

It can be seen from (\ref{hlambda}) that the order of magnitude of
$h(\la)$ is the larger of the magnitude of
$\tilde{Z}_{xx}-\tilde{Z}_{yy}$ and of $\tilde{Z}_{xy}$. The latter is
always $1/\la^{d(x,y)}$, while the former is by definition
$1/\la^{co(x,y)+1}$. There is tunneling only if $co(x,y)+1 \geq d$
hence the magnitude of $h$ is always governed by the $1/\la^d$ term,
hence the magnitude of
the tunneling time is $1/h(Q) \approx Q^{d(x,y)-1}$
as claimed. \qed

\subsection{Triple wells}

Let now $G$ be a vertex-transitive graph and let the three wells be denoted by $x,y,z$. This time we have to analyze the eigenfunctions of a 3-by-3 matrix $\tilde{Z}(\la)$ instead of a 2-by-2. Let us denote the pairwise distances between the wells by $a = d(x,y) \leq b = d(x,z) \leq c = d(y,z)$. Henceforth we shall assume that the wells are pairwise $2a$-cospectral. Though the general case could be done along the same lines, we confine ourselves to the study of sufficiently cospectral wells mainly in order to reduce the number of cases to check. Already in this special case we are able to exhibit interesting phenomena. 
As we shall see, there are different type of behaviors depending on the distances between the wells: (i) $a < b  \leq c$, (ii) $a = b < c$,  (iii)  $a = b = c$. Before looking at the particular cases, we can further simplfy the matrix $\tilde{Z}$ a little.

% At first glance one might think that the cospectrality assumption implies that the diagonal entries are all equal up to the order of  This is not the case, as in the definition of $P_{vv}(k)$ we only consider paths that avoid all the wells. 
%Let us denote by $P'_{vv}(k)$ the same qunatity but without the constraint of avoiding the other wells. Then, of course for small $k$ we have $P'_{vv}(k) = P_{vv}(k)$. In fact, the smallest $k$ for which the two quantities are different is for $k = 2\min(d(v,w) : {w \in \{x,y,z\} \setminus \{v\}})$. Let us write 
%\begin{equation}
%Z'_{vv}(\la) = \sum_{k=1}^{\infty} \frac{P'_{vv}(k)}{(1-\la)^k}. 
%\end{equation}

%Then $\tilde{Z}_{vv}(\la) - Z'_{vv}(\la)$ is a power series in $1/\la$ with leading term $(1/\la)^k: k = 2\min(d(v,w) : {w \in \{x,y,z\} \setminus \{v\}})$. Of course by the symmetry of the graph $Z'_{vv} = Z'_{ww}$ for any $v,w \in X$. 
If we subtract $\tilde{Z}_{xx}(\la)$ from each diagonal entry in $\tilde{Z}$ then we only change the eigenvalues, but not the eigenvectors of the matrix. By the cospectrality assumption $\tilde{Z}_{yy} - \tilde{Z}_{xx} = O(1/(1-\la)^{2a}) = \tilde{Z}_{zz} - \tilde{Z}_{xx}$. The modified matrix for the triple-well case looks like this (remembering that $a \leq b \leq c$):

\begin{equation}\label{fmatrix}
\tilde{Z} - \tilde{Z}_{xx} I
%\left ( \begin{array}{ccc}
%f_{2a}(\la) & f_a(\la) & f_b(\la) \\
%f_a(\la) & f_{2a}(\la) & f_c(\la) \\
%f_b(\la) & f_c(\la) & f_{2b}(\la)
%\end{array} \right)
= \frac{1}{(1-\la)^a} \cdot
\left ( \begin{array}{ccc}
0 & f_0(\la) & f_{b-a}(\la) \\
f_{0}(\la) & f_{a+a'}(\la) & f_{c-a}(\la) \\
f_{b-a}(\la) & f_{c-a}(\la) & f_{a+a''}(\la)
\end{array} \right)
\end{equation}
where  $a', a'' \geq 0$ and each entry denotes a power series, the index indicating the first term that may have non-zero coefficient: $f_m(t) = \sum_{k = m}^{\infty} c_k /(1-\la)^k$, such that $c_m \neq 0$. By abuse of notation, even if there are equalities among the indices, the corresponding power series are allowed to be different. For large $\la$ the eigenvalues and eigenvectors of this matrix can be approximated by taking the limit of the matrix on the right hand side. Let 
\begin{equation}
M = \lim_{\la \to \infty} 
\left ( \begin{array}{ccc}
0 & f_0(\la) & f_{b-a}(\la) \\
f_{0}(\la) & f_{a+a'}(\la) & f_{c-a}(\la) \\
f_{b-a}(\la) & f_{c-a}(\la) & f_{a+a''}(\la)
\end{array} \right)
\end{equation}
 denote the limit of this matrix as $\la \to \infty$. Note that $M$ depends only on the graph and the position of the wells. Let $\mu_1,\mu_2,\mu_3$ denote the three eigenvalues of $M$ and $\psi_1,\psi_2,\psi_3$ the corresponding eigenvectors. 

\begin{lemma}\label{eigenvectorcondition} Assume that $M$ has three distinct eigenvalues. Then a neccesary condition for perfect asymptotic tunneling from  $y$ to $z$ is that  $\psi_3(x) = \psi_3(y) + \psi_3(z) = 0 = \psi_1(y) - \psi_1(z)  = \psi_2(y)-\psi_2(z)$ for some permutation of the eigenvectors.
\end{lemma}

\begin{proof}
Assume that there is perfect asymptotic tunneling from $y$ to $z$. Since $M$ has no multiple eigenvalues, the three relevant eigenvectors of the Hamiltonian restricted to the wells converge to the three well-defined eigendirections of $M$. Using the definition of perfect tunneling and compactness it is easy to see that we get three real numbers $r_1, r_2, r_3$ and three unit complex numbers $\rho_1,\rho_2,\rho_3$ such that $(0,1,0) = \sum r_j \psi_j$ and $(0,0,1) = \sum \rho_j r_j \psi_j$. 

If either of the $r_j$'s would be 0, then (assuming $r_1$ is zero) we get that $(0,\rho_2,-1) = (\rho_2 - \rho_3)r_3 \psi_3$. Since $\psi_3$ is real, this can only be if $\rho_2 = \pm 1$. Since obviously $\rho_3 \neq \rho_2$ we may assume wlog that $\rho_2 \neq -1$, hence $\rho_2 = 1$ and $\psi_3 = c(0,1,-1)$. This proves the part of the statment involving $\psi_3$, the rest follows by orthogonality of the eigenvectors.

Now we assume neither of the $r_j$'s is 0, and denote $\Psi_j = r_j\psi_j$. 
Next we show that $\Psi_j(x)$ has to be 0 for at least one $j$. Assume this is not true. Then $\sum \Psi_j(x) = 0$ and $\sum \rho_j \Psi_j(x) = 0$. From this we get that $(\rho_1 - \rho_2)\Psi_2(x) + (\rho_1 - \rho_3)\Psi_3(x) = 0$. Since none of the $\Psi_j(x)$'s are zero, this implies that $\rho_1$ is a convex combination of $\rho_2$ and $\rho_3$. For unit complex numbers this can only be if they are all equal -- a contradiction. 

So we may assume $\Psi_3(x) = \psi_3(x) = 0$. But then $\Psi_1(x) = -\Psi_2(x) \neq 0$ and $\rho_1 \Psi_1(x)  = -\rho_2 \Psi_2(x)$, hence $\rho_1 = \rho_2$.  Looking at $y$ we get that $\Psi_1(y)+\Psi_2(y) + \Psi_3(y) = 1$ while  $\rho_1(\Psi_1(y)+\Psi_2(y)) + \rho_3 \Psi_3(y) = 0$. This can only happen if $\rho_1 = -\rho_3 =\pm 1$.  We may again wlog assume that $\rho_1 = 1$ and $\rho_3 = -1$. Finally taking the difference $(0,1-1) = (0,1,0) - (0,0,1) = 2r_3\psi_3$ we get that $\psi_3(z) + \psi_3(y) = 0$. The rest of the statement follows from orthogonality again.
\end{proof}

Whether the above condition of is also sufficient for perfect tunneling is described in the next lemma.

\begin{lemma}\label{gammacondition}
Assume $M$ has three distinct eigenvalues and that $\psi_3(x) = \psi_3(y) + \psi_3(z) = 0$ and that neither $\psi_1(x)$ nor $\psi_2(x)$ is 0.  Let 
\begin{equation}
\gamma = \frac{\mu_1 - \mu_3}{\mu_1 - \mu_2}.
\end{equation}
\begin{enumerate}
\item If $\gamma = p/q$ is rational where $p \in \Z$ is odd while $q \in \Z$ is even, then there is perfect tunneling from $y$ to $z$ in time $O(Q^{a-1})$.
\item If $\gamma$ is irrational, then there is perfect tunneling from $y$ to $z$, but the tunneling time $f$ has to satisfy $Q^{a-1}  = o(f(Q))$.
\item Otherwise we cannot say for sure that there is perfect tunneling from $y$ to $z$. However if there is, the tunneling time $f$ has to be satisfy $Q^{a} = O(f(Q))$.
\end{enumerate}
\end{lemma}

\begin{proof}
Let $\nu$ be an eigenvalue of $\tilde{Z}$ with eigenvector $\psi$. Let us recall that $\psi$ is the restriction of an eigenfunction of the Hamiltonian to the wells with eigenvalue $\la$ if and only if $Q = (1-\la)(1-\nu)$. If we denote by $\nu_j = \nu_j(\la) : j=1,2,3$ the eigenvalues of the matrix on the right hand side of (\ref{fmatrix}) then on one hand $\mu_j = \lim_{\la \to \infty} \nu_j$, on the other hand the eigenvalues of $\tilde{Z}$ are $\nu_j / (1-\la)^a + \tilde{Z}_{xx}(\la)$ and the solutions for $Q$ are given by 
\begin{equation} \label{Q_jeq}
Q_j = (1-\la)(1-\tilde{Z}_{xx}(\la)) - \nu_j /(1-\la)^{a-1} : j=1,2,3.
\end{equation} Denoting the inverse functions by $\la_j(Q) : j = 1,2,3$ we see that the three relevant eigenvalues of the Hamiltonian are precisely the $\la_j(Q)$'s.

Since $\psi_1(x)$ and $\psi_2(x)$ are non-zero, all three eigenvectors participate in the decomposition of the pure state $y$. Then from the proof of the previous lemma we can see that perfect tunneling takes places from $y$ to $z$ in time $f(Q)$ if and only if for large $Q$ there is a $t \in [0,f(Q)]$ and a unit complex number $\rho$ such that $e^{it\la_1(Q)} \approx e^{it\la_2(Q)} \approx \rho$ and $e^{it\la_3(Q)} \approx -\rho$. (Here $\la_3(Q)$ corresponds to the eigenvector $\psi_3$ in the previous proof.) More precisely when
\begin{equation}\label{tunnelingeq}
\limsup_{Q \to \infty} \inf_{t \in [0,f(Q)]} |e^{it(\la_1(Q)-\la_2(Q))} - 1|^2 + |e^{it(\la_1(Q)-\la_3(Q))} + 1|^2  =0
\end{equation}
Let us introduce the lattice $\mathcal{L} = \{(p,q) \in \Z^2: p \equiv 0 (2), q \equiv 1 (2)\}$ and denote by
\begin{equation}
\Gamma = \{ (t \frac{\la_1(Q) - \la_2(Q)}{\pi} , t \frac{\la_1(Q) - \la_3(Q)}{\pi}) \in \R^2  :  t \in [0,f(Q)] \}
\end{equation}
the relevant curve traced out in $\R^2$ as $t$ runs over the interval $[0,f(Q)]$. Then (\ref{tunnelingeq}) is equivalent to 
\begin{equation}\label{LGammaeq}
\limsup_{Q \to \infty} \mbox{dist}(\mathcal{L}, \Gamma) = 0.
\end{equation}

Let us analyze what happens to the curve $\Gamma$ as $Q \to \infty$.  Using (\ref{Q_jeq}) together with Lemma~\ref{inverselemma} we can write
\begin{equation}\label{laj_lak}
\la_{j_1}(Q) - \la_{j_2}(Q) = (1+O(1/Q))(\nu_{j_1}(Q) - \nu_{j_2}(Q))/(1-Q)^{a-1}.
\end{equation}
If $f(Q) = o(Q^a)$ then clearly $\Gamma$ converges uniformly to a segment through the origin whose slope is $\gamma$ and whose length is of magnitude $f(Q)/Q^{a-1}$. 
First of all, if $f(Q) = o(Q^{a-1})$ then $\Gamma$ converges uniformly to a single point, namely the origin. As $(0,0) \not \in \mathcal{L}$, in this case there can not be perfect tunneling.
\begin{enumerate}
\item If $\gamma = q/p$ is rational where $q \in \Z$ is odd while $p \in \Z$ is even, then the line with slope $\gamma$ passes through a point of $\mathcal{L}$. Hence there is a constant $K$ such that if $f(Q) = KQ^{a-1}$ then the limit segment of $\Gamma$ will pass through $\mathcal{L}$ hence (\ref{LGammaeq}) will hold. Thus there is perfect tunneling in time $O(Q^{a-1})$. 
\item If $\gamma$ is irrational, then the line of slope $\gamma$ will pass arbitrarily close to $\mathcal{L}$. If $f(Q) = O(Q^{a-1})$ then the length of the limiting segment will be bounded hence the $\limsup$ in (\ref{LGammaeq}) will be strictly positive and there will not be perfect tunneling. On the other hand if $Q^{a-1} = o(f(Q))$ then $\Gamma$ will arbitrarliy approximate any finite piece of the limiting line, hence there will be perfect asymptotic tunneling. 
\item Finally, if $\gamma$ is rational but not odd-over-even, then the lattice will be separated from the limiting line. Hence if $f(Q) = o(Q^{a})$ then the $\limsup$ in (\ref{LGammaeq}) will be positive again. Thus tunneling cannot take place within time $o(Q^{a})$.
\end{enumerate}
\end{proof}

\begin{remark}\label{conjremark}
The only place where we use that $M$ has no multiple eigenvalues is that it implies that the eigenfunctions converge to the three eigenvectors of $M$. If $M$ does have multiple eigenvalues but we can compute the eigenvectors of $\tilde{Z}$, and denote their limits by $\psi_{1,2,3}$, then the statement of the Lemma remains true, as the same argument works.

It is not hard to see, that if $(\la_1(Q) - \la_3(Q)) / (\la_1(Q) - \la_2(Q))$ is not a constant as a function of $Q$, then even in the last case there will always be perfect tunneling and the tunneling time will depend on the degree of the first non-zero term in $(\la_1(Q) - \la_3(Q)) / (\la_1(Q) - \la_2(Q)) - \gamma$. On the other hand if this is a constant function and $\gamma$ is rational but not odd-over-even, then there is no perfect tunneling at all. 
\end{remark}

\begin{proof}[Proof of Theorem~\ref{triplewelltheorem}] We have to consider two cases:
\paragraph{$\bf a < b \leq c:$}

In this case 
\begin{equation}
M = 
\left ( \begin{array}{ccc}
0 & c_1 & 0 \\
c_1 & 0 & 0 \\
0 & 0 & 0
\end{array} \right),
\end{equation}
hence the eigenvectors of $\tilde{Z}$ converge to $(1,1,0)/\sqrt{2} ; (1,-1,0)/\sqrt{2} ; (0,0,1)$ with respective eigenvalues $c_1, -c_1, 0$. This of course means that there is complete asymptotic tunneling between the two nearest wells, and there is no tunneling to or from the third well. Further we get that the tunneling time is asymptotically $\pi/|\la_1(Q) - \la_2(Q)| = 2c_1 \pi / (1-Q)^{a-1}$.

\paragraph{$\bf a = b < c:$}

In this case  
\begin{equation}
M = \left ( \begin{array}{ccc}
0 & c_1 & c_2 \\
c_1 & 0 & 0 \\
c_2 & 0 & 0
\end{array} \right).
\end{equation} 
Let us denote $r= \sqrt{c_1^2+c_2^2}$. 
Then $M$'s  eigenvectors are 
\begin{equation}
\begin{array}{l}
\psi_1 = \frac{(-r,c_1,c_2)}{\sqrt{2}r}; \\ \psi_2 = \frac{(r,c_1,c_2)}{\sqrt{2}r};\\ \psi_3 = \frac{(0,-c_2,c_1)}{r}.
\end{array}
\end{equation} The corresponding eigenvalues are $-r, r$ and $0$. Some computation from (\ref{fmatrix}) shows that $\nu_1(\la) \sim 1/(1-\la)^{c-a}$. 
So if the particle starts from state $x$ then, since $(1,0,0)$ is a multiple of $(\psi_2 - \psi_3)$, we get that the particle tunnels to the mixed state in which the probability of being in $y$ and $z$ is both $1/2$.  

The situation is different if we start the particle from $y$. By Lemma~\ref{eigenvectorcondition} it follows that if $c_1 \neq c_2$ then there cannot be complete tunneling from $y$ to $z$. So let us assume $c_1 = c_2$.  Since $r\neq 0$, all conditions of Lemma~\ref{gammacondition} are satisfied and $|\gamma| = 1/2$, hence part a) applies and we get that there is perfect tunneling in time $\Theta(Q^{a-1})$.
\end{proof}

\subsection{Equilateral triple wells}

We are still using all the above notations, in particular we are analyzing $\tilde{Z}$ using (\ref{fmatrix}). Now we have 
\begin{equation}
M = 
\left ( \begin{array}{ccc}
0 & c_{xy} & c_{xz} \\
c_{xy} & 0 & c_{yz} \\
c_{xz} & c_{yz} & 0
\end{array} \right),
\end{equation}
where the $c$'s are defined according to (\ref{cuveq}) and 
 depend on the particular geodesics connecting the wells.
 
 \begin{proof}[Proof of Theorem~\ref{equilateraltheorem}]

  All of the $c$'s are strictly positive, hence $M$ has no eigenvectors with a single non-zero entry. This implies that some partial tunneling always happens. The question is how to determine when perfect tunneling happens. It is easy to see that $M$ is non-singular, hence its eigenvalues are non-zero reals. The characteristic polynomial is $x^3 - (c_{xy}^2+c_{xz}^2+c_{yz}^2)x - 2c_{xy}c_{xz}c_{yz}$. If this has multiple roots then it has a common root with its derivative and just looking at the signs, this common root can only be $x = - \sqrt{(c_{xy}^2+c_{xz}^2+c_{yz}^2)/3}$. But that implies $\sqrt[3]{c_{xy} c_{xz} c_{yz}} = \sqrt{c_{xy}^2+c_{xz}^2+c_{yz}^2/3}$, which can only happen if $c_{xy} = c_{xz} = c_{yz}$ since the $c$'s are non-negative.

Let us first assume this is not the case. Then there are no multiple eigenvalues hence there are three distinct eigenvectors, which are the limits of the corresponding eigenvectors of the right hand side of (\ref{fmatrix}). As before, we denote the three eigenvectors by $\psi_1, \psi_2, \psi_3$. 

Clearly, the condition of Lemma~\ref{eigenvectorcondition} for $\psi_3$ is equivalent to $c_{xy} = c_{xz}$. This proves part a) of the theorem. It is also not difficult to compute that the eigenvalues are $\mu_{1,2} = \frac{c_{yz} \pm \sqrt{c_{yz}^2 + 8c_{xz}^2}}{2}, \mu_3 = -c_{yz}$. Thus in this case 
\begin{equation}
\gamma = \frac{ 3c_{yz} + \sqrt{c_{yz}^2 + 8c_{xz}^2}}{2 \sqrt{c_{yz}^2 + 8c_{xz}^2}}
\end{equation}
Since the $c$'s are rational numbers by construction, $\gamma$ is rational if and only if $\sqrt{c_{yz}^2  + 8c_{xz}^2}$ is rational.
\begin{claim}
If $\gamma$ is rational it cannot be odd-over-even.
\end{claim}
\begin{proof}
By multiplying through with the appropriate integer, we may assume that $c_{yz}, c_{xz}$ are coprime integers. If $c_{yz}$ is odd then so is $\sqrt{c_{yz}^2  + 8c_{xz}^2}$, hence the numerator is even and the denominator is twice an odd number, hence $\gamma$ cannot be odd-over-even. So we may assume $c_{yz} = 2p$, hence $c_{xz}$ is odd (as they are coprime). But then $2\sqrt{p^2 + 2c_{xz}}$ cannot be an integer, since $p^2 + 2c_{xz} $ is 2 or 3 modulo 4, which means it cannot be a full square. This is a contradiction hence the claim is true.
\end{proof}

Lemma~\ref{gammacondition} now implies parts b) and c) of the theorem, except for the case when $c_{xy} = c_{xz} = c_{yz}$. In this last case the difficulty is that the eigenvectors of $M$ are not well-defined, hence the arguments where we use that the eigenvectors of $\tilde{Z}$ converge to the eigenvectors of $M$ do not work automatically. This can be overcome by choosing a sequence of $Q$'s along which the eigenvectors of $\tilde{Z}$ do converge, and denoting their limits (which will of course still be pairwise orthogonal eigenvectors of $M$) by $\psi_1, \psi_2, \psi_3$. Now by the arguments of Lemma~\ref{eigenvectorcondition} we still get that $\psi_3$ has to be parallel to $(0,1,-1)$. Also we know that $\psi_1$ is parallel to $(1,1,1)$ since this is a 1-dimensional eigenspace of $M$. And then $\psi_2$ is parallel to $(-2,1,1)$. Now the corresponding eigenvalues are $2, -1, -1$ and the arguments of Lemma~\ref{gammacondition} go through word by word. As $\gamma = 1$ in our case, and 1 is not odd-over-even, we get that the tunneling time has to satisfy $Q^a = O(f(Q))$ as before. This completes the proof of part c).
\end{proof}

\begin{remark}
Now we see that by Remark~\ref{conjremark} the only way Conjecture~\ref{conj} could be false is if there existed a graph with three of its vertices $x,y,z$ forming an equilateral triangle such that $c_{xy} = c_{xz} \neq c_{yz}$ and $\sqrt{c_{yz}^2 + 8c_{xz}^2}$ is rational, and at the same time the ratio $(\la_1(Q) - \la_3(Q)) / (\la_1(Q) - \la_2(Q))$ is independent of $Q$. Though currently we are unable to prove it, we believe such graphs do not exist.
\end{remark}

\subsection{Instability}

\begin{lemma} If there is an order three symmetry of the graph that permutes the wells (in particular $c_{xy} = c_{xz} = c_{yz}$) then there is no perfect tunneling from $y$ to $z$, in fact $TC(y,z) = 4/9$.
\end{lemma}

\begin{proof}
Using the order three symmetry of the graph it is easy to see that $\tilde{Z}_{xx} = \tilde{Z}_{yy} = \tilde{Z}_{zz}$ and $\tilde{Z}_{xy} = \tilde{Z}_{yz} = \tilde{Z}_{zx}$. Hence $(1,1,1)$ is an eigenvector of $\tilde{Z}$, and the whole subspace orthogonal to $(1,1,1)$ is a 2-dimensional eigenspace. This way in the decomposition of $(1,0,0)$ there will only be two terms: $1/3(1,1,1) + 1/3(2,-1,-1)$. As before, these two eigenvectors are rotating at different speeds, so after certain amount time (not hard to see that the order of magnitude is $Q^{a-1}$) their phases will be opposite, and the state of the particle becomes $1/3(-1,2,2)$, which will be the extremally tunneled state. Hence $TC(y,z) = 4/9$. It is also obvious that  at any moment during the whole evolution the states $y$ and $z$ will be equally likely. 
\end{proof}

\begin{proof}[Proof of Theorem~\ref{instabthm}]
Let us consider the following simple construction. Let our graph consist of three paths joined at a common vertex $O$. (So $O$ has degree 3, the endpoints of the paths have degree 1 and all the other vertices have degree 2.) Let the three wells be the three neighbors of $O$. Let the length of the paths going through $x,y,z$ be denoted respectively by $a,b,c$. If $a=b=c$ then by the previous lemma we get that for any two wells the tunneling coefficient is $4/9$. 

On the other hand if $a  \neq b = c$ then there is perfect tunneling from $y$ to $z$. To see this, first we have to understand the eigenvectors of $\tilde{Z}$ Let us assume $a < b$, though the $a > b$ case would work just as well. It is easy to see that $\tilde{Z}_{xy} = \tilde{Z}_{xz} = \tilde{Z}_{yz}$  and by symmetry $\tilde{Z}_{yy} = \tilde{Z}_{zz}$. Since up to length $2a-1$ the closed paths from all three wells are identical, but the paths of length $2a$ are not, $0 \neq \tilde{Z}_{xx} - \tilde{Z}_{yy} ~ 1/\la^{2a}$. Thus $\tilde{Z}$ is of the form
\begin{equation}
\tilde{Z} = 
\left( \begin{array}{ccc}
p & q & q \\
q & r & q \\
q & q & r
\end{array} \right)
\end{equation}
where $p \neq r$. Easy computation shows that the eigenfunctions of this matrix are $(c_1, 1,1); (c_2,1,1); (0,1,-1)$ and $c_1 \neq 0 \neq c_2$. It is also not hard to compute that all three eigenvalues are distinct. The eigenvalues of the corresponding $M$ matrix are $1,-1,-1$ and hence $\gamma = 1$. Thus, by Remark~\ref{conjremark} in order to be able to apply Lemma~\ref{gammacondition} and conclude that there is indeed perfect tunneling from $y$ to $z$, we have to check that $(\la_1 - \la_3)/(\la_1 - \la_2)$ is not a constant function of $Q$. If it were constant, it would be equal to its limit as $Q \to \infty$, which is precisely $\gamma$. But since $\gamma = 1$, this meant that $\la_2 = \la_3$ for all $Q$, which means that the inverse functions are also equal, and hence for the eigenvalues of $\tilde{Z}$ we get that $\nu_2(\la) = \nu_3(\la)$. But this contradicts our previous observations. Hence there is perfect tunneling from $y$ to $z$.
\end{proof}

\end{document}